\documentclass[conference, final, letterpaper]{IEEEtran}

\usepackage{hyperref}
\usepackage{ifpdf}
\usepackage{cite}
\ifpdf
    \usepackage[pdftex]{graphicx}
    \graphicspath{{./fig/}}
    \DeclareGraphicsExtensions{.pdf}
\else
    \usepackage[dvips]{graphicx}
    \graphicspath{{./fig/}}
    \DeclareGraphicsExtensions{.eps}
\fi
\usepackage{bbm}
\usepackage{theorem}
\usepackage[cmex10]{amsmath}
\usepackage{amssymb}
\usepackage{stmaryrd}
\usepackage{esdiff}

\newcommand{\erasure}{\ensuremath{\vartimes}}

\newcommand{\F}{\ensuremath{\mathbb{F}}}
\newcommand{\N}{\ensuremath{\mathbb{N}}}
\renewcommand{\vec}[1]{\ensuremath{\mathbf{#1}}}

\newcommand{\coloneq}{\ensuremath{\mathrel{\mathop:}=}}

\newcommand{\pub}{\ensuremath{\underline{p}}}
\newcommand{\pob}{\ensuremath{\overline{p}}}

\newcommand{\T}[2]{\ensuremath{T_{#2}^{(#1)}}}
\newcommand{\Tz}[1]{\ensuremath{\T{z}{#1}}}
\newcommand{\Tzo}{\ensuremath{\T{z}{1}}}
\newcommand{\Tzk}{\ensuremath{\T{z}{k}}}
\newcommand{\Tzz}{\ensuremath{\T{z}{z}}}

\newcommand{\estimate}[1]{\tilde{#1}}
\newcommand{\compared}[1]{\widehat{#1}}
\newcommand{\C}{\ensuremath{\mathcal{C}}}
\newcommand{\innerobj}[1]{\ensuremath{#1^\mathrm{i}}}
\newcommand{\Cin}{\innerobj{\C}}
\newcommand{\nin}{\innerobj{n}}
\newcommand{\kin}{\innerobj{k}}
\newcommand{\din}{\innerobj{d}}
\newcommand{\Rin}{\innerobj{R}}
\newcommand{\iin}{\innerobj{\vec{a}}}
\newcommand{\cin}{\innerobj{\vec{c}}}
\newcommand{\rin}{\innerobj{\vec{r}}}
\newcommand{\ein}{\innerobj{\vec{e}}}
\newcommand{\cestin}{\innerobj{\estimate{\vec{c}}}}
\newcommand{\iestin}{\innerobj{\estimate{\vec{a}}}}
\newcommand{\decin}{\innerobj{\mathrm{dec}}}
\newcommand{\outerobj}[1]{\ensuremath{#1^\mathrm{o}}}
\newcommand{\Cout}{\outerobj{\C}}
\newcommand{\nout}{\outerobj{n}}
\newcommand{\kout}{\outerobj{k}}
\newcommand{\dout}{\outerobj{d}}

\newcommand{\iout}{\outerobj{\vec{a}}}
\newcommand{\cout}{\outerobj{\vec{c}}}
\newcommand{\rout}{\outerobj{\vec{r}}}

\newcommand{\rcompout}{\outerobj{\compared{\vec{r}}}}

{\theoremstyle{plain}
   }

{\theoremstyle{plain}
   \newtheorem{theorem}{Theorem}}

{\theoremstyle{plain}
   \newtheorem{lemma}{Lemma}}

{\theoremstyle{plain}
   \newtheorem{corollary}{Corollary}}

\begin{document}

\title{Optimal Threshold--Based Multi--Trial Error/Erasure Decoding with the Guruswami--Sudan Algorithm}

\IEEEoverridecommandlockouts

\author{
\authorblockN{Christian Senger, Vladimir R. Sidorenko, Martin Bossert}\thanks{This work has been supported by DFG,
Germany, under grant BO~867/21-1. Vladimir Sidorenko is on leave from IITP, Russian Academy of Sciences, Moscow, Russia.}
\authorblockA{\small Inst. of Telecommunications and Applied Information Theory\\
Ulm University, Ulm, Germany \\
\{christian.senger$\;\vert\;$vladimir.sidorenko$\;\vert\;$martin.bossert\}@uni-ulm.de}
\and
\authorblockN{Victor V. Zyablov}
\authorblockA{\small Inst. for Information Transmission Problems\\
Russian Academy of Sciences, Moscow, Russia \\
zyablov@iitp.ru}
}

\maketitle

\begin{abstract}
Traditionally, multi--trial error/erasure decoding of {\em Reed--Solomon (RS)} codes is based on {\em Bounded Minimum Distance (BMD)} decoders with an erasure option. Such decoders have error/erasure tradeoff factor $\lambda=2$, which means that an error is twice as expensive as an erasure in terms of the code's minimum distance. The {\em Guruswami--Sudan (GS)} list decoder can be considered as state of the art in algebraic decoding of RS codes. Besides an erasure option, it allows to adjust $\lambda$ to values in the range $1<\lambda\leq2$. Based on previous work \cite{senger_sidorenko_bossert_zyablov:2010b}, we provide formulae which allow to optimally (in terms of residual codeword error probability) exploit the erasure option of decoders with arbitrary $\lambda$, if the decoder can be used $z\geq 1$ times. We show that BMD decoders with $z_\mathrm{BMD}$ decoding trials can result in lower residual codeword error probability than GS decoders with $z_\mathrm{GS}$ trials, if $z_\mathrm{BMD}$ is only slightly larger than $z_\mathrm{GS}$. This is of practical interest since BMD decoders generally have lower computational complexity than GS decoders.
\end{abstract}

\section{Introduction}\label{sec:intro}

{\em Multi--trial error/erasure (MTEE)} decoding or {\em Generalized Minimum Distance (GMD)} decoding  \cite{forney:1966b, forney:1966a} is a technique which applies multiple decoding trials of an error/erasure decoder on each received word, each time with a different number of erased most unreliable symbols. The ideas behind this approach are to not let unreliable received symbols interfere the decoding process and to exhaustively try the set of most promising erasure patterns. MTEE decoding performs surprisingly well, especially when the channel is in good shape. This is naturally the case when we consider concatenated codes. Here, the inner code and the channel can be considered jointly as a {\em super channel} which, due to the inner decoder's error--correcting capabilities, has low symbol error probability.

We investigate a particular concatenated code construction which is widely used in practice and standards, e.g.  the Consultative Committee for Space Data System's (CCSDS) Telemetry Channel \cite{ccsds:2002}. In this construction, the inner code is a convolutional code with a {\em Maximum Likelihood (ML)} decoder. The outer code is a traditional {\em Reed--Solomon (RS)} code. We stress that the inner code needs to be {\em tailbited} to insulate channel error events to single symbols of the outer received word.

Our target is to minimize the residual codeword error probability after decoding. We consider {\em threshold erasing}, which means that each output of the inner ML decoder is measured against a set of $z\geq 1$ real thresholds $\Tzo, \ldots, \Tzz$, $\Tzo\leq\cdots \leq \Tzz$. If the reliability of the symbol (which is an output of the inner ML decoder) falls below threshold $\Tzk$ in decoding trial $k$, $1\leq k\leq z$, then the symbol is erased and replaced by the {\em erasure marker} $\erasure$. The threshold erasing method dates back to Blokh and Zyablov \cite{blokh_zyablov:1982} and is different from the {\em symbol erasing} method used in Forney's original work about GMD decoding. There, the received symbols are ordered according to their reliabilities and an increasing number of most unreliable received symbols is erased in each of the $z$ decoding trials. 

Currently, the most powerful technique for algebraic decoding of RS codes is the {\em Guruswami--Sudan (GS)} list decoder \cite{guruswami_sudan:1999}. It can be parametrized to obtain error/erasure tradeoff factors $\lambda$ in the range $1<\lambda\leq 2$. $\lambda$ expresses the relative cost of errors compared to erasures in terms of required {\em Hamming distance}. Generally, increasing the {\em multiplicity parameter} $\nu$ brings along higher list decoding radius, increased decoding complexity, and smaller $\lambda$. We will elaborate the latter fact in the course of the paper.

The GS decoder has been extended to a soft--input algorithm by K\"otter and Vardy in their award--winning 2003 paper \cite{koetter_vardy:2003}. Their algorithm is based on setting the multiplicity of each interpolation point in the GS decoder according to the reliability of the corresponding received symbol. Another promising approach for soft--input decoding of RS codes has recently been published by Nguyen et. al. \cite{nguyen_pfister_narayanan:2011} and is based on {\em rate--distortion theory}. In our work, we investigate the potential of threshold erasing, when the outer code is decoded in multiple trials with the GS decoder. The results are based on our previous papers \cite{senger_sidorenko_zyablov:2009b, senger_sidorenko_bossert_zyablov:2010b}, in which we consider outer BMD decoding ($\lambda=2$) of {\em Bose--Chaudhuri--Hocquenghem} codes and outer decoding of {\em Interleaved Reed--Solomon (IRS)} codes ($\lambda=(\ell+1)/\ell$, $\ell\in\N\setminus\{0\}$), respectively.

For the sake of completeness we should also mention other publications on related topics, e.g. maximization of the decoding radius of concatenated block codes with an outer $\lambda$--decoder using threshold erasing \cite{senger_sidorenko_bossert_zyablov:2008a, senger_sidorenko_bossert_zyablov:2010a} and symbol erasing \cite{sidorenko_senger_bossert_zyablov:2008, sidorenko_chaaban_senger_bossert:2009, sidorenko_senger_bossert_zyablov:2010}. Outer list decoders have already been considered by Nielsen \cite{nielsen:2001}, but with the aim of maximizing the decoding radius of the concatenated code construction. An overview of the different erasing techniques with an arbitrary number of decoding trials is given in \cite{weber_abdel-ghaffar:2003}.

The rest of the paper is organized as follows. In Section~\ref{sec:conc}, we describe structure and threshold--based MTEE decoding of the aforementioned concatenated code construction. We use results from \cite{senger_sidorenko_zyablov:2009b, senger_sidorenko_bossert_zyablov:2010b} to derive optimal threshold locations for outer decoding with $1<\lambda\leq 2$ in Section~\ref{sec:thresholds}. Here and in the rest of the paper, {\em optimal} means {\em minimizing the residual codeword error probability}. Section~\ref{sec:nonconstant} deals with the GS decoder's non--constant $\lambda$ and shows how our result from Section~\ref{sec:thresholds} can be applied nevertheless. Optimal threshold locations are used in Section~\ref{sec:sim} to plot error probability curves of an exemplary concatenated code.
It will turn out, that for the considered setting the high--complexity GS decoder is in many cases not worth the effort and multiple trials of low--complexity BMD decoding yield comparable or even lower residual codeword error probabilities. 
We conclude our paper in Section~\ref{sec:conclusions}.

\section{MTEE Decoding of Concatenated Codes}\label{sec:conc}

A concatenated code $\C(n, k, d)$ consists of an {\em inner code} $\Cin(\F_2; \nin, \kin=m, \din)$ and an {\em outer code} $\Cout(\F_{2^m}; \nout, \kout, \dout)$. The resulting concatenated code $\C$ is binary and, w.l.o.g., we restrict ourselves to this most practical case.

The information vector $\iout\in\F_{2^m}^{\kout}$ is encoded into an outer codeword $\cout\coloneq(\outerobj{c}_0, \ldots, \outerobj{c}_{\nout-1}) \in\Cout\subseteq\F_{2^m}^{\nout}$ of the outer code. Each $2^m$--ary symbol $\outerobj{c}_j$, $j=0, \ldots, \nout-1$, can be interpreted as a binary vector $\iin_j\in\F_2^{\kin}$ of length $m$. These vectors serve as information for the inner code and are encoded into inner codewords $\cin_j\in\Cin\subseteq\F_2^{\nin}$. Arranging the $\cin_j$ as columns of a matrix gives the codeword matrix of the concatenated code $\C$, which is transmitted over a {\em binary symmetric channel (BSC)} with crossover probability $p$.

The receiver obtains erroneous columns $\rin_j\coloneq\cin_j+\ein_j$, which are fed into the ML decoder for $\Cin$. It returns inner codeword estimates $\cestin_j\coloneq\decin(\rin_j)$. The information parts $\iestin_j$ are extracted from the $\cestin_j$ and mapped to symbols $\outerobj{r}_j\in\F_{2^m}$. The resulting vector $\rout\coloneq(\outerobj{r}_0, \ldots, \outerobj{r}_{\nout-1})$ is the input for the MTEE decoder of $\Cout$. The MTEE decoder performs  erasing with the threshold set $\mathcal{T}\coloneq\left\{\Tz1, \ldots, \Tzz\right\}$, \mbox{$\Tzo\leq\cdots \leq \Tzz$}. It calculates a {\em reliability value} $v_j$ for every received symbol $\outerobj{r}_j\triangleq\iestin_j$ according to
\begin{equation*}
  v_j\coloneq
    \frac{1}{\nin}%
      \ln\left(%
        \frac{\Pr\left(\rin_j|\cestin_j\right)}{\sum_{\cin\in\Cin\setminus\left\{\cestin_{\scriptscriptstyle{j}}\right\}} \Pr\left(\rin_j|\cin\right)}%
      \right)
\end{equation*}
and $\mathcal{T}$ is applied in the following manner:
\begin{equation*}
  \outerobj{\compared{r}}_{k, j}\coloneq\left\{\begin{array}{ll}
  \outerobj{r}_j, & \mathrm{if}\;v_j\geq \Tzk\\
  \erasure, & \mathrm{if}\;v_j<\Tzk
  \end{array}\right.,
\end{equation*}
$k=1, \ldots, z$. Note that the particular calculation of the reliability value stems from \cite[Corollary to Theorem 1]{forney:1968} and results in decision regions which minimize both the error- and the error--or--erasure probability of the outer decoder at the same time. Result of the erasing procedure is the {\em input list} $\mathcal{I}\coloneq\left\{\rcompout_1, \ldots, \rcompout_z\right\}$, in which \mbox{$\rcompout_k\coloneq(\outerobj{\compared{r}}_{k, 0}, \ldots, \outerobj{\compared{r}}_{k, \nout-1})$}. Each element of the input list is fed into the outer decoder with \mbox{$1<\lambda\leq 2$} and multiplicity $\nu$. Since we allow the outer decoder to be a list decoder, each decoding trial potentially returns a result list $\rho_k$. These lists are merged into the {\em overall} result list $\mathcal{R}\coloneq\bigcup_{k=1}^z \rho_k$. We have a decoding success whenever $\cout\in\mathcal{R}$.

\section{Optimal Thresholds Locations}\label{sec:thresholds}

As a starting point for our derivation of the optimal threshold locations we generalize \cite[Theorem 1]{senger_sidorenko_bossert_zyablov:2010b}.

Several cases are possible when a single received symbol $\outerobj{r}_j$, which could be either correct ($\outerobj{r}_j=\outerobj{c}_j$) or erroneous (\mbox{$\outerobj{r}_j=\outerobj{c}_j+\outerobj{e}_j$}), is considered. First, the symbol might be correct and never erased by any threshold. We denote the probability of this event by
\begin{equation*}
  p_r := \Pr(\outerobj{r}_j=\outerobj{c}_j\text{ and never erased}).
\end{equation*}
Second, the symbol might be erroneous and never erased, the probability of this event is
\begin{equation*}
  p_l := \Pr(\outerobj{r}_j\neq\outerobj{c}_j\text{ and never erased}).
\end{equation*}
Third, the symbol might be erased by every threshold in $\mathcal{T}$, in this case we do not distinguish whether it is correct or not and denote the probability by
\begin{equation*}
  p_c := \Pr(\outerobj{r}_j\text{ always erased}).
\end{equation*}
The last two cases are for correct and erroneous symbols that are not erased by thresholds $\Tzo, \ldots, \Tz{k}$, but erased by all (larger) thresholds $\Tz{k+1}, \ldots, \Tzz$. The corresponding probabilities are
\begin{align*}
  \pob_k &:= \Pr\left(\outerobj{r}_j=\outerobj{c}_j\text{ and erased by}\;\Tz{k+1}\; \text{but not by}\; \Tz{k}\right)\\
  \pub_k &:= \Pr\left(\outerobj{r}_j\neq\outerobj{c}_j\text{ and erased by}\;\Tz{k+1}\; \text{but not by}\; \Tz{k}\right).
\end{align*}
It is clear that these probabilities must sum up to one, i.e. $p_r+p_l+p_c+\sum_{k=1}^{z-1} (\pob_k+\pub_k)=1$.

Since it is similar to the derivation of \cite[Theorem 1]{senger_sidorenko_bossert_zyablov:2010b}, we omit the generalized derivation here and immediately state the following theorem.

\begin{theorem}\label{thm:conditions}
If the outer decoder has error/erasure tradeoff factor $\lambda$, $1<\lambda\leq 2$, and can correct up to (including) $\delta$ erasures, then the following conditions are necessary and sufficient for an optimal MTEE threshold set \mbox{$\mathcal{T}=\left\{\Tzo, \ldots, \Tzz\right\}$}.
\begin{align*}
  p_l^\frac{1}{\lambda} &= p_c,\\
  p_c &= (\pub_1^\frac{1}{\lambda-1}\, \pob_1)^{1-\frac{1}{\lambda}},
\end{align*}
and
\begin{equation*}
  \forall\,k=1, \ldots, z-2: \pub_k^\frac{1}{\lambda-1}\, \pob_k=\pub_{k+1}^\frac{1}{\lambda-1}\, \pob_{k+1}.
  \end{equation*}
For $\mathcal{T}$ fulfilling these conditions, the residual codeword error probability $P_e$ can be approximated by
\begin{multline}
  P_e\approx p_l^\frac{\delta}{\lambda}= p_c^{\delta}=(\pub_1^\frac{1}{\lambda-1}\, \pob_1)^{\delta(1-\frac{1}{\lambda})}=\cdots\\
  \cdots = (\pub_z^\frac{1}{\lambda-1}\, \pob_z)^{\delta(1-\frac{1}{\lambda})}.\label{eqn:Pe}
\end{multline}
\end{theorem}

In case of BMD- and many other decoders $\delta=\dout-1$. However, we will see later that for the GS decoder we might also require smaller values of $\delta$.

Following \cite{blokh_zyablov:1982}, we state simple approximations for the probabilities of Theorem~\ref{thm:conditions} in our previous paper \cite{senger_sidorenko_bossert_zyablov:2010b}. We repeat them in Lemma~\ref{lemma:approx} to clarify the further derivation of the optimal threshold set. The lemma is based on  spherical approximations of the inner code's Voronoi cells and the exponential error bounds for erasure schemes derived by Forney \cite{forney:1968}, which generalize Gallager's error exponents for the BSC \cite{gallager:1965}.

\begin{lemma}[Senger et. al. \cite{senger_sidorenko_bossert_zyablov:2010b}]\label{lemma:approx}
Simple approximations of the probabilities $p_c, p_l, \pob_k,$ and $\pub_k$ are given by
\begin{align*}
  p_c &\approx \exp\left(-\left(E_0(\Rin)-s\,\Tzo\right)\nin\right), \\
  p_l &\approx \exp\left(-\left(E_0(\Rin)+s\,\Tzz\right)\nin\right), \\
  \pob_k &\approx \exp\left(-\left(E_0(\Rin)-s\,\Tz{k+1}\right)\nin\right), \\
  \pub_k &\approx \exp\left(-\left(E_0(\Rin)+s\,\Tz{k}\right)\nin\right),
\end{align*}
$k=1, \ldots, z-1$, where $E_0(\Rin)$ is Gallager's error exponent for ML decoding of a code with rate $\Rin$ and transmission over a BSC. $s$, $0< s\leq 1/2$, is the corresponding optimization parameter.
\end{lemma}

The conditions from Theorem~\ref{thm:conditions} and the approximations from Lemma~\ref{lemma:approx} allow to obtain analytic formulae for the optimal threshold locations. Their number $z$, the rate $\Rin$ of the inner code and $\lambda$ are parameters. Inserting the approximations into the conditions results in the following system of $z$ recurrent equations.
\begin{align}
  p_l^\frac{1}{\lambda} &= p_c\Longleftrightarrow\nonumber\\
  \frac{1}{\lambda}\left(%
    E_0(\Rin)+s\,\Tzz%
  \right)%
  &= E_0(\Rin)-s\,\Tzo\label{eqn:eqn1},
\end{align}
\begin{align}
  p_c &= (\pub_1^\frac{1}{\lambda-1}\, \pob_1)^{1-\frac{1}{\lambda}}\Longleftrightarrow\nonumber\\
  (\lambda+1)\Tz{1} &= (\lambda-1)\Tz{2}\label{eqn:eqn2},
\end{align}
and, $\forall\,k=1,\ldots, z-2$,
\begin{align}
  \pub_k^\frac{1}{\lambda-1}\, \pob_k &= \pub_{k+1}^\frac{1}{\lambda-1}\, \pob_{k+1}\Longleftrightarrow\nonumber\\
  \frac{1}{\lambda-1} \left(\lambda\Tz{k+1}-\Tz{k}\right) &= \Tz{k+2}\label{eqn:eqn3}.
\end{align}
Equations (\ref{eqn:eqn1}), (\ref{eqn:eqn2}), and (\ref{eqn:eqn3}) allow to prove our main theorem.

\begin{theorem}\label{thm:locationBD}
The optimal threshold set $\mathcal{T}=\left\{\Tzo, \ldots, \Tzz\right\}$ for MTEE decoding of a concatenated code with an inner ML decoder and an outer decoder with error/erasure tradeoff factor $\lambda$, $1<\lambda< 2$, is given by
\begin{equation*}\label{eqn:TzkBD}
  \Tzk:=\frac{E_0(\Rin)}{s}\cdot F(\lambda),
\end{equation*}
where $E_0(\Rin)$ is Gallager's error exponent for the BSC, $s$ is the corresponding optimization parameter, $0<s\leq\frac{1}{2}$, and
\begin{equation}\label{eqn:F}
  F(\lambda):=%
  \frac{%
    2 \left(\frac{1}{\lambda-1}\right)^{k-1}-\lambda
  }{%
    2 \left(\frac{1}{\lambda-1}\right)^{z-1}-\lambda
  }.
\end{equation}

\end{theorem}

\begin{proof}
The statement follows from the unique solution of the recurrence relation (\ref{eqn:eqn1}), (\ref{eqn:eqn2}), and (\ref{eqn:eqn3}) for $1<\lambda<2$.
\end{proof}

\begin{corollary}\label{cor:locationBMD}
For outer BMD decoding, i.e. $\lambda=2$, the optimal threshold set is given by
\begin{equation*}
  \Tzk:=\frac{E_0(\Rin) (2k-1)}{s(2z+1)}.
\end{equation*}
\end{corollary}

\begin{proof}
The statement follows from the unique solution of the recurrence relation (\ref{eqn:eqn1}), (\ref{eqn:eqn2}), and (\ref{eqn:eqn3}) for $\lambda=2$.
\end{proof}

Corollary~\ref{cor:locationBMD} coincides with a result of Blokh and Zyablov \cite{blokh_zyablov:1982}. Thus, we obtain their threshold location formula as a special case of our main Theorem~\ref{thm:locationBD}.

Fig.~\ref{fig:thresholds} shows the optimal threshold sets for $z=20$, $\Rin=1/2$, $p=0.02$, and $\lambda=1.1, \ldots, 1.9, 2.0$. Each line represents one threshold set, Darker color of the curve means larger $\lambda$. The optimal threshold set for outer BMD decoding ($\lambda=2$, see Corollary~\ref{cor:locationBMD}) is given as a reference. Note that $F(\lambda)$ is constant for fixed $\lambda$ and $z$, other crossover  probabilities $p$ of the BSC simply scale the threshold locations by a factor.

\begin{figure}[htbp]
\centering
\includegraphics[width=212pt, clip]{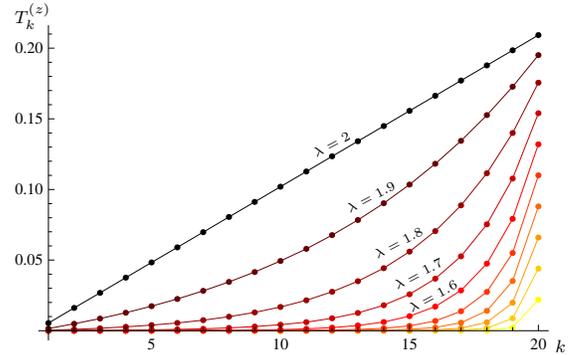}
\caption{Optimal threshold sets according to Theorem~\ref{thm:locationBD} and Corollary~\ref{cor:locationBMD} for $z=20$, $\Rin=1/2$, $p=0.02$, and $\lambda=1.1, \ldots, 1.9, 2.0$.}
\label{fig:thresholds}
\end{figure}

It is easy to prove that $\Tzk$ is non-increasing with decreasing $\lambda$, a fact which can also be observed in Fig.~\ref{fig:thresholds}. This means that with decreasing $\lambda$, the number of erased symbols generally becomes smaller. We could have expected such a behavior since with decreasing $\lambda$, the relative cost of errors decreases and thereby also the effect of erasing unreliable received symbols.

We can use Theorem~\ref{thm:conditions} to obtain an approximation of the residual codeword error probability after MTEE decoding with an optimal threshold set obtained by Theorem~\ref{thm:locationBD}. To do so, we use the second term from (\ref{eqn:Pe}) and write
\begin{equation*}
  P_e\approx p_l^\frac{\delta}{\lambda}.
\end{equation*}
Inserting the approximation of $p_l$ from Lemma~\ref{lemma:approx} gives
\begin{equation*}
  P_e\approx \left(\exp\left(-\left(E_0(\Rin)+s\,\Tzz\right)\nin\right)\right)^\frac{\delta}{\lambda},
\end{equation*}
in which we can replace $\Tzz$ as given by Theorem~\ref{thm:locationBD} for $1<\lambda<2$ or Corollary~\ref{cor:locationBMD} for $\lambda=2$, respectively. This results in the following theorem and its corollary.

\begin{theorem}\label{thm:Pe}
The residual codeword error probability of MTEE decoding of a concatenated code with an inner ML decoder, an outer decoder with error/erasure tradeoff factor $\lambda$, $1<\lambda< 2$, maximal number of correctable erasures $\delta$, and an optimal threshold set $\mathcal{T}=\left\{\Tzo, \ldots, \Tzz\right\}$  can be approximated by
\begin{equation*}
  P_e\approx
  \exp\left(%
    -2 E_0(\Rin)\delta\,\frac{%
      \left(\frac{1}{\lambda-1}\right)^z-1
    }{%
    2\left(\frac{1}{\lambda-1}\right)^z-\lambda
    }\,\nin
  \right).
\end{equation*}
\end{theorem}

\begin{corollary}
For traditional outer BMD decoding, i.e. $\lambda=2$ and $\delta=\dout-1$, we have the approximation
\begin{equation*}
  P_e\approx \exp\left(-2E_0(\Rin)(\dout-1)\,\frac{z}{2z+1}\,\nin\right).
\end{equation*}

\end{corollary}

So far, we assumed that $\lambda$ is constant for any number of erased symbols. This is true for BMD decoders but not for the GS decoder as we will see in the following section

\section{Dealing with the GS List Decoder's Non--Constant $\lambda$}\label{sec:nonconstant}

The {\em decoder capability function} ({\em DCF}, a constraint on the number $\tau$ of erasures and the number $\varepsilon$ of errors, that can be corrected concurrently) of a BMD decoder is 
\begin{equation*}
  \nout-\tau-2\varepsilon>\kout-1.
\end{equation*}
For $\tau$ erasures, $0\leq\tau\leq\dout-1$, the decoder fails to correct $\varepsilon_\mathrm{BMD}(\tau)\coloneq(\nout-\kout+1-\tau)/2$ or more errors. The indeed constant $\lambda$ for any number of erasures is given by the negative reciprocal value of $\varepsilon_\mathrm{BMD}(\tau)$'s slope, i.e.
\begin{equation*}
  \lambda_\mathrm{BMD}\coloneq-\left(\diff{\varepsilon_\mathrm{BMD}(\tau)}{\tau}\right)^{-1}=2.
\end{equation*}
The situation is different for the GS decoder. For simplicity, we restrict ourselves to the best (in terms of achievable list decoding radius) case, i.e. multiplicity $\nu\rightarrow\infty$. It's DCF is
\begin{equation*}
    \frac{(n-\tau-\varepsilon)^2}{n-\tau}>k-1,
\end{equation*}
resulting in $\varepsilon_\mathrm{GS}(\tau)\coloneq\nout-\tau-\sqrt{(\kout-1)(\nout-\tau)}$ and
\begin{align*}
  \lambda_\mathrm{GS}(\tau) &\coloneq%
  -\left(\left.\diff{\varepsilon_\mathrm{GS}(t)}{t}\right|_{t=\tau}\right)^{-1}\\
  &=\left(1-\frac{k-1}{2\sqrt{(k-1)(n-\tau)}}\right)^{-1},
\end{align*}
which is a strictly monotonic increasing function of $\tau$ and thereby not usable in Theorem~\ref{thm:locationBD}. We will now show that near--optimal threshold locations for the GS decoder can be calculated using Theorem~\ref{thm:locationBD}.

It is straightforward to see that for any $\tau$, a decoder with radius $\varepsilon_\mathrm{GS}(\tau)$ can be transformed into a decoder with radius $\varepsilon'_\mathrm{GS}(\tau)<\varepsilon_\mathrm{GS}(\tau)$ by simply discarding all decoding results with $\tau$ erasures and $\varepsilon\geq \varepsilon'_\mathrm{GS}(\tau)$ errors. This fact and the monotonicity of $\lambda(\tau)$ allow to conclude that any tangent of $\varepsilon_\mathrm{GS}(\tau)$ at $\tau=\kappa$, $0\leq\kappa\leq\dout-1$, specifies a {\em tangent decoder} with radius 
\begin{equation*}
  \varepsilon_{\mathrm{GS}, \kappa}(\tau)\coloneq \varepsilon_\mathrm{GS}(\kappa)+\left.\diff{\varepsilon_\mathrm{GS}(t)}{t}\right|_{t=\kappa}(\tau-\kappa)
\end{equation*}
and constant error/erasure tradeoff factor
\begin{align*}
  \lambda_{\mathrm{GS}, \kappa} &=-\left(\left.\diff{\varepsilon_\mathrm{GS}(t)}{t}\right|_{t=\kappa}\right)^{-1}\\
  &=\left(1-\frac{\kout-1}{2\sqrt{(\kout-1)(\nout-\kappa)}}\right)^{-1},
\end{align*}
that can be imitated by the GS list decoder. Its maximum number of correctable erasures $\delta_{\mathrm{GS}, \kappa}$ is obtained by solving $\varepsilon_{\mathrm{GS}, \kappa}(\tau)=0$ for $\tau$ and taking the floor of the result.

Since $\lambda_{\mathrm{GS}, \kappa}$ is independent of $\tau$, Theorems~\ref{thm:locationBD} and \ref{thm:Pe} can be applied with $\lambda_{\mathrm{GS}, \kappa}$ and $\delta_{\mathrm{GS}, \kappa}$ to obtain optimal threshold locations and residual codeword error probabilities for tangent decoders which can be imitated by the GS decoder. The optimal tangent decoder is determined by
\begin{equation}\label{eqn:optkappa}
  \kappa^*\coloneq
  \arg\min_{0\leq\kappa\leq\dout-1}\left\{
    -\delta_{\mathrm{GS}, \kappa}\,\frac{%
      \left(\frac{1}{\lambda_{\mathrm{GS}, \kappa}-1}\right)^z-1
    }{%
    2\left(\frac{1}{\lambda_{\mathrm{GS}, \kappa}-1}\right)^z-\lambda_{\mathrm{GS}, \kappa}
    }
  \right\},
\end{equation}
which is independent of the the ML error exponent. Thus, tangent decoders determined by (\ref{eqn:optkappa}) are optimal for all BSC crossover probabilities.

\section{Simulation Results -- Traditional BMD Decoding can Beat the GS List Decoder}\label{sec:sim}

\begin{figure*}[htbp]
\centering
\includegraphics[width=490pt, clip]{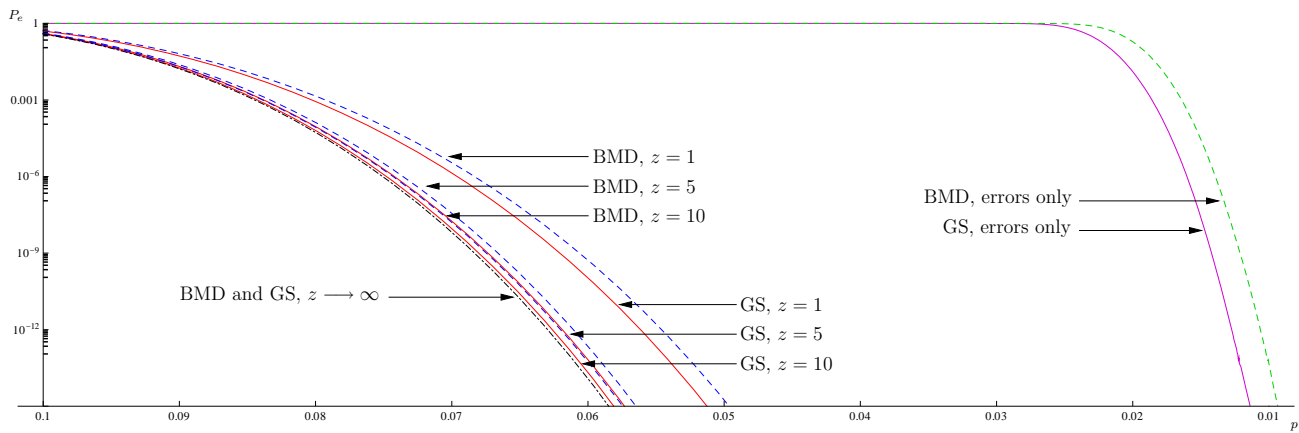}
\caption{Residual codeword error probability curves for $\Cout(\F_{2^8}; 255, 144, 112)$ and $\Rin=1/2$.}
\vspace{-0.3cm}
\label{fig:sim}
\end{figure*}

Let us consider the outer RS code $\Cout(\F_{2^8}; 255, 144, 112)$ with GS decoder. We consider $z_\mathrm{GS}\coloneq 1, 5, 10$ outer decoding trials. Based on (\ref{eqn:optkappa}), Table~\ref{tab:tangentdecoders} states the parameters of the corresponding optimal tangent decoders.
\renewcommand{\arraystretch}{1.2}
\begin{table}[htbp]
\vspace{-0.2cm}
\begin{center}
\begin{tabular}{c|c|c|c}
$z_\mathrm{GS}$ & $\kappa^*$ & $\lambda_{\mathrm{GS}, \kappa^*}$ & $\delta_{\mathrm{GS}, \kappa^*}$\\
\hline\hline
$1$ & $41$ & $1.69126$ & $107$\\
$5$ & $72$ & $1.79208$ & $110$\\
$10$ & $85$ & $1.84699$ & $111$
\end{tabular}
\end{center}
\caption{Optimal tangent decoders for $z_\mathrm{GS}\coloneq 1, 5, 10$.}
\label{tab:tangentdecoders}
\vspace{-0.7cm}
\end{table}
\renewcommand{\arraystretch}{1}

As inner code, we assume a tailbited rate $\Rin=1/2$ convolutional code with ML decoder. This allows to use Theorem~\ref{thm:Pe} in order to plot the solid red residual codeword error curves for outer GS decoding in Fig.~\ref{fig:sim}. Additionally, we consider outer BMD decoding and allow the decoder to be run $z_\mathrm{BMD}\coloneq 1,5, 10$ times (dashed blue curves). We observe that the gain of tangent decoding diminishes for growing $z$. Since both residual codeword error probabilities (optimal tangent decoder and BMD decoder) converge to the same value, i.e. 
\begin{equation*}
P_e\stackrel{z\rightarrow\infty}{\longrightarrow}\exp(-E_0(\Rin)(\dout-1)\nin)
\end{equation*}
(dash--dotted black curve), we conclude that for every number $z_\mathrm{GS}$ of outer GS decoding trials, there exists a number $z_\mathrm{BMD}\geq z_\mathrm{GS}$ of outer BMD decoding trials that achieves either the same or lower residual codeword error probability. This allows to trade a number of high--complexity GS decoding trials for a (generally larger) number of low--complexity BMD decoding trials, extending the options of the system designer.

\section{Conclusions}\label{sec:conclusions}

We generalized our results from \cite{senger_sidorenko_bossert_zyablov:2010b} to the case of arbitrary error/erasure tradeoff factors $\lambda$ in the range $1<\lambda\leq 2$. We derived formulae for optimal thresholds applicable in MTEE decoding, our generalization allows to use the GS list decoder for the outer code. Based on our derivation, we gave approximations of the residual codeword error probability after outer decoding for the full range of $\lambda$. This allowed to compare outer GS list decoding with traditional, low--complexity, BMD decoding. Our main result is that for the particular concatenated coding scheme under consideration (outer RS code, inner convolutional code with ML decoding, e.g. used in\cite{ccsds:2002} ), $z_\mathrm{BMD}$ trials of outer BMD decoding can outperform $z_\mathrm{GS}$ trials of GS decoding if $z_\mathrm{BMD}\geq z_\mathrm{GS}$. This is interesting for practical applications, since BMD decoders have low computational complexity and are widely deployed.

\def\noopsort#1{}


\begin{thebibliography}{10}
\providecommand{\url}[1]{#1}
\csname url@samestyle\endcsname
\providecommand{\newblock}{\relax}
\providecommand{\bibinfo}[2]{#2}
\providecommand{\BIBentrySTDinterwordspacing}{\spaceskip=0pt\relax}
\providecommand{\BIBentryALTinterwordstretchfactor}{4}
\providecommand{\BIBentryALTinterwordspacing}{\spaceskip=\fontdimen2\font plus
\BIBentryALTinterwordstretchfactor\fontdimen3\font minus
  \fontdimen4\font\relax}
\providecommand{\BIBforeignlanguage}[2]{{%
\expandafter\ifx\csname l@#1\endcsname\relax
\typeout{** WARNING: IEEEtran.bst: No hyphenation pattern has been}%
\typeout{** loaded for the language `#1'. Using the pattern for}%
\typeout{** the default language instead.}%
\else
\language=\csname l@#1\endcsname
\fi
#2}}
\providecommand{\BIBdecl}{\relax}
\BIBdecl

\bibitem{senger_sidorenko_bossert_zyablov:2010b}
\BIBentryALTinterwordspacing
C.~Senger, V.~R. Sidorenko, M.~Bossert, and V.~V. Zyablov, ``Optimal thresholds
  for {GMD} decoding with $\frac{\ell+1}{\ell}$--extended bounded distance
  decoders,'' in \emph{Proc. IEEE Int. Symp. on Inform. Theory}, Austin, TX,
  USA, June 2010, pp. 1100--1104. [Online]. Available:
  \url{http://dx.doi.org/10.1109/ISIT.2010.5513698}
\BIBentrySTDinterwordspacing

\bibitem{forney:1966b}
G.~D. Forney, ``{G}eneralized {M}inimum {D}istance decoding,'' \emph{IEEE
  Trans. Inform. Theory}, vol. IT-12, pp. 125--131, April 1966.

\bibitem{forney:1966a}
------, \emph{Concatenated Codes}.\hskip 1em plus 0.5em minus 0.4em\relax
  Cambridge, MA, USA: M.I.T. Press, 1966.

\bibitem{ccsds:2002}
\emph{Telemetry Channel Coding}, Consultative Committee for Space Data Systems,
  October 2002, recommendation for Space Data System Standards, CCSDS
  101.0-B-6, Blue Book, Issue 6.

\bibitem{blokh_zyablov:1982}
E.~L. Blokh and V.~V. Zyablov, \emph{Linear Concatenated Codes}.\hskip 1em plus
  0.5em minus 0.4em\relax Nauka, 1982, in Russian.

\bibitem{guruswami_sudan:1999}
\BIBentryALTinterwordspacing
V.~Guruswami and M.~Sudan, ``Improved decoding of {Reed-Solomon} and
  algebraic-geometric codes,'' \emph{IEEE Trans. Inform. Theory}, vol. IT-45,
  no.~6, pp. 1755--1764, September 1999. [Online]. Available:
  \url{http://dx.doi.org/10.1109/18.782097}
\BIBentrySTDinterwordspacing

\bibitem{koetter_vardy:2003}
\BIBentryALTinterwordspacing
R.~Koetter and A.~Vardy, ``Algebraic soft-decision decoding of {Reed--Solomon}
  codes,'' \emph{IEEE Trans. Inform. Theory}, vol. IT-49, no.~11, pp.
  2809--2825, November 2003. [Online]. Available:
  \url{http://dx.doi.org/10.1109/TIT.2003.819332}
\BIBentrySTDinterwordspacing

\bibitem{nguyen_pfister_narayanan:2011}
\BIBentryALTinterwordspacing
P.~S. Nguyen, H.~D. Pfister, and K.~R. Narayanan, ``{On Multiple Decoding
  Attempts for Reed--Solomon Codes},'' \emph{IEEE Trans. Inform. Theory}, vol.
  IT-57, no.~2, pp. 668--691, February 2011. [Online]. Available:
  \url{http://dx.doi.org/10.1109/TIT.2010.2095202}
\BIBentrySTDinterwordspacing

\bibitem{senger_sidorenko_zyablov:2009b}
\BIBentryALTinterwordspacing
C.~Senger, V.~R. Sidorenko, and V.~V. Zyablov, ``On {G}eneralized {M}inimum
  {D}istance decoding thresholds for the {AWGN} channel,'' in \emph{Proc. XII
  Symposium Problems of Redundancy in Information and Control Systems}, St.
  Petersburg, Russia, May 2009, pp. 155--163. [Online]. Available:
  \url{http://k36.org/redundancy2009/proceedings.pdf}
\BIBentrySTDinterwordspacing

\bibitem{senger_sidorenko_bossert_zyablov:2008a}
\BIBentryALTinterwordspacing
C.~Senger, V.~R. Sidorenko, M.~Bossert, and V.~V. Zyablov, ``Decoding
  generalized concatenated codes using interleaved {R}eed--{S}olomon codes,''
  in \emph{Proc. IEEE Int. Symp. on Inform. Theory}, Toronto, ON, Canada, July
  2008. [Online]. Available: \url{http://dx.doi.org/10.1109/ISIT.2008.4595300}
\BIBentrySTDinterwordspacing

\bibitem{senger_sidorenko_bossert_zyablov:2010a}
------, ``Multi-trial decoding of concatenated codes using fixed thresholds,''
  \emph{Problems of Information Transmission}, vol.~46, no.~2, pp. 127--141,
  2010.

\bibitem{sidorenko_senger_bossert_zyablov:2008}
\BIBentryALTinterwordspacing
V.~R. Sidorenko, C.~Senger, M.~Bossert, and V.~V. Zyablov, ``Single-trial
  adaptive decoding of concatenated codes,'' in \emph{Proc. International
  Workshop on Algebraic and Combinatorial Coding Theory}, Pamporovo, Bulgaria,
  June 2008. [Online]. Available:
  \url{http://www.moi.math.bas.bg/acct2008/b44.pdf}
\BIBentrySTDinterwordspacing

\bibitem{sidorenko_chaaban_senger_bossert:2009}
\BIBentryALTinterwordspacing
V.~R. Sidorenko, A.~Chaaban, C.~Senger, and M.~Bossert, ``On extended
  {F}orney--{K}ovalev {GMD} decoding,'' in \emph{Proc. IEEE Int. Symp. on
  Inform. Theory}, Seoul, Korea, July 2009. [Online]. Available:
  \url{http://dx.doi.org/10.1109/ISIT.2009.5205900}
\BIBentrySTDinterwordspacing

\bibitem{sidorenko_senger_bossert_zyablov:2010}
\BIBentryALTinterwordspacing
V.~R. Sidorenko, C.~Senger, M.~Bossert, and V.~V. Zyablov, ``Single--trial
  decoding of concatenated codes using fixed or adaptive erasing,''
  \emph{Advances in Mathematics of Communications (AMC)}, vol.~4, no.~1, pp.
  49--60, February 2010. [Online]. Available:
  \url{http://dx.doi.org/10.3934/amc.2010.4.49}
\BIBentrySTDinterwordspacing

\bibitem{nielsen:2001}
R.~R. Nielsen, ``Decoding concatenated codes using {Sudan's} algorithm,'' 2001.

\bibitem{weber_abdel-ghaffar:2003}
\BIBentryALTinterwordspacing
J.~H. Weber and K.~A.~S. Abdel-Ghaffar, ``Reduced {GMD} decoding,'' \emph{IEEE
  Trans. Inform. Theory}, vol. IT-49, no.~4, pp. 1013--1027, April 2003.
  [Online]. Available: \url{http://dx.doi.org/10.1109/TIT.2003.809504}
\BIBentrySTDinterwordspacing

\bibitem{forney:1968}
G.~D. Forney, ``Exponential error bounds for erasure, list, and decision
  feedback schemes,'' \emph{IEEE Trans. Inform. Theory}, vol. IT-14, pp.
  206--220, March 1968.

\bibitem{gallager:1965}
R.~G. Gallager, ``A simple derivation of the coding theorem and some
  applications,'' \emph{IEEE Trans. Inform. Theory}, vol. IT-11, pp. 3--18, Jan
  1965.

\end{thebibliography}
\end{document}